\documentclass[submission,copyright,creativecommons]{eptcs}
\def\publicationstatus{This is a revision of a paper with the same title,\\
  published in EPTCS 451, 2026, pp. 33--47,\\
  doi:\href{https://doi.org/10.4204/EPTCS.451.3}{10.4204/EPTCS.451.3}}
\providecommand{\event}{AFL 2026}
\usepackage{amsmath,amssymb,amsfonts,amsthm,thmtools}
\usepackage{array,setspace,multicol,multirow}
\usepackage{hyperref,enumerate,verbatim,graphicx,url}
\usepackage{tikz}
\usetikzlibrary{calc}
\usetikzlibrary{positioning,automata,fit,arrows.meta}
\tikzset{arrows={[scale=1.1]}}
\tikzset{every edge/.style={draw,->,>=Stealth,auto}}
\tikzset{every node/.append style={minimum size=1.5cm}}
\tikzset{every state/.append style={minimum size=1.5cm}}
\usepackage{anyfontsize,lmodern,xspace}
\usepackage{cleveref}
\newcommand{\Cerny}{\v{C}ern{\'y}\xspace}
\newtheorem{theorem}{Theorem}
\newtheorem{corollary}[theorem]{Corollary}
\newtheorem{lemma}[theorem]{Lemma}
\newtheorem{proposition}[theorem]{Proposition}
\newtheorem{conjecture}[theorem]{Conjecture}
\newtheorem{openproblem}[theorem]{Open Problem}

\DeclareSymbolFont{rsfscript}{OMS}{rsfs}{m}{n}
\DeclareSymbolFontAlphabet{\mathrsfs}{rsfscript}
\DeclareMathOperator{\lspan}{span}

\usepackage{iftex}
\ifpdf
  \usepackage{underscore}         
  \usepackage[T1]{fontenc}        
\else
  \usepackage{breakurl}           
\fi
\renewcommand{\O}{\mathcal{O}}

\title{Synchronizing Automata:\\Open Problems\footnote{This work was supported by the National Science Centre, Poland under project number 2021/41/B/ST6/03691.}}
\author{Marek Szyku{\l}a
\institute{Institute of Computer Science, University of Wroc{\l}aw, Wroc{\l}aw, Poland}
\email{msz@cs.uni.wroc.pl}}
\def\titlerunning{Synchronizing Automata: Open Problems}
\def\authorrunning{M. Szyku{\l}a}
\begin{document}
\maketitle

\begin{abstract}
We survey selected open problems in the theory of synchronizing automata, centered around the famous \v{C}ern{\'y} conjecture.
A deterministic finite automaton is called \emph{synchronizing} if it admits a \emph{reset word} whose action maps all states to a single state.
The \v{C}ern{\'y} conjecture states that every synchronizing automaton with $n$ states possesses a reset word of length at most $(n-1)^2$.
We discuss avoiding words, compressing a state with another, synchronization of a (given or any) subset, complexity of deciding the synchronizability, average reset threshold, and linear-algebraic methods.
Some new auxiliary results are also presented.
\end{abstract}
\section{Introduction}

We survey some open problems from the core of the theory of synchronizing automata, and we present their state-of-the-art.
We take the perspective focused on general (unrestricted) automata rather than on particular subclasses with additional assumptions.
The paper also includes some new auxiliary and motivational results, which emerged from attempts of solving the problems.

We refer to excellent recent Volkov's surveys:
The general survey around the \Cerny conjecture \cite{Volkov2022Survey} and the list of subclasses of automata for which a better upper bound than the general one is known \cite{Volkov2026Survey}.

The first four problems seem to be closely related to each other and to the \Cerny problem; they also seem to be of similar difficulty.
Hence, progress in any of them will likely be helpful in solving the others.
We conclude them with a generalization.

In the second part, we discuss the fundamental computational problem of checking whether an automaton is synchronizing, the average reset threshold (where recently a breakthrough result has appeared \cite{CP2025ShortSynchronizingWordsForRandomAutomata}), and some ideas for bypassing limitations of the linear-algebraic methods (which are widely applicable in the theory of synchronizing automata).

Most of the concepts touched can be visually illustrated and interactively studied using the \emph{SynchroViewer} tool \cite{SynchroViewer}.

\subsection{Preliminaries}

We consider deterministic finite complete semiautomata $\mathrsfs{A}=(Q,\Sigma,\delta)$, where $Q$ is the set of \emph{states}, $\Sigma$ is the \emph{alphabet}, and $\delta\colon Q \times \Sigma \to Q$ is the \emph{transition function}, extended naturally to a function $Q \times \Sigma^* \to Q$.
Throughout the paper, by $n$ we always denote the number of states $|Q|$.

The \emph{image} of a subset $S \subseteq Q$ by the action of a word $w \in \Sigma^*$ is $\delta(S,w) = \{\delta(q,w) \mid q \in S\}$, and its \emph{preimage} is $\delta^{-1} = \{q \in Q \mid \delta(q,w) \in S\}$.

A word $w \in \Sigma^*$ is \emph{reset} (or \emph{synchronizing}) if $|\delta(Q,w)|=1$, i.e., every state is mapped by the action of $w$ to the same state.
An automaton is \emph{synchronizing} if it admits a reset word.

We say that a word $w$ \emph{compresses} a (necessarily nonempty) subset $S \subseteq Q$ if $|\delta(S,w)|<|S|$, and it \emph{extends} a subset $S \subsetneq Q$ if $|\delta^{-1}(S,w)|>|S|$.
Also, we say that a word $w$ \emph{synchronizes} a subset $S \subseteq Q$ if $|\delta(S,w)|=1$.

An automaton is \emph{strongly connected} if for every pair of states $p,q \in Q$, there is a word $w$ with $\delta(p,w) = q$.
When considering synchronization problems, the most interesting cases are strongly connected synchronizing automata, as the \Cerny problem can be reduced to them.

\section{The Problems}

\subsection{The \Cerny Conjecture}

The famous \Cerny conjecture is the following:

\begin{conjecture}[\Cerny]
If an $n$-state automaton is synchronizing, then it has a reset word of length at most $(n-1)^2$.
\end{conjecture}

The bound is met for every $n$ by the \Cerny automata \cite{Cerny1964} (\Cref{fig:Cerny}).

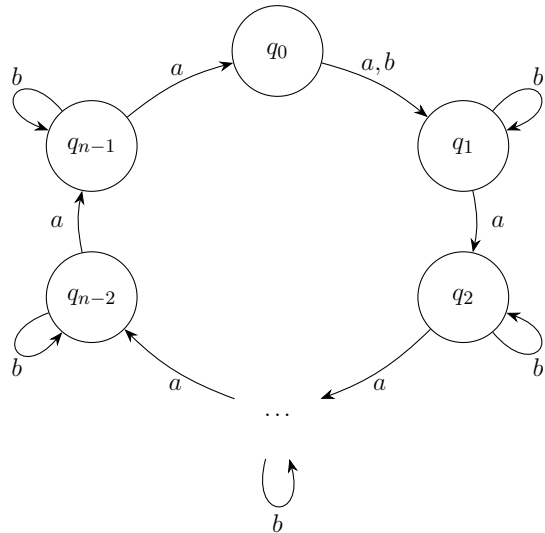
\begin{figure}[htb]\large\centering\begin{tikzpicture}[node distance=2cm,scale=0.8,every node/.style={transform shape}]
\tikzset{every loop/.style={min distance=.5cm,looseness=6}}
\node[state] (q0) {$q_0$};
\node[state] [below right=.5cm and 2cm of q0] (q1) {$q_1$};
\node[state] [below right=3cm and 2cm of q0] (q2) {$q_2$};
\node[state,draw=none] [below=4.5cm of q0] (qdots) {$\ldots$};
\node[state] [below left=3cm and 2cm of q0] (qn-2) {$q_{n-2}$};
\node[state] [below left=.5cm and 2cm of q0] (qn-1) {$q_{n-1}$};

\draw[->,>=Stealth] (q0) to[bend left=12] node[midway,above]{$a,b$} (q1);
\draw[->,>=Stealth] (q1) to[bend left=12] node[midway,above,right=.1cm]{$a$} (q2);
\draw[->,>=Stealth] (q2) to[bend left=12] node[midway,below]{$a$} (qdots);
\draw[->,>=Stealth] (qdots) to[bend left=12] node[midway,below]{$a$} (qn-2);
\draw[->,>=Stealth] (qn-2) to[bend left=12] node[midway,above,left=.1cm]{$a$} (qn-1);
\draw[->,>=Stealth] (qn-1) to[bend left=12] node[midway,above]{$a$} (q0);

\draw[->,>=Stealth](q1) to[in=20,out=50,looseness=7] node[midway,above]{$b$} (q1);
\draw[->,>=Stealth](q2) to[in=-20,out=-50,looseness=7] node[midway,below]{$b$} (q2);
\draw[->,>=Stealth](qdots) to[in=-75,out=-105,looseness=7] node[midway,below]{$b$} (qdots);
\draw[->,>=Stealth](qn-2) to[in=-130,out=-160,looseness=7] node[midway,below]{$b$} (qn-2);
\draw[->,>=Stealth](qn-1) to[in=-200,out=-230,looseness=7] node[midway,above]{$b$} (qn-1);
\end{tikzpicture}
\caption{The \Cerny automaton with $n$ states.}\label{fig:Cerny}
\end{figure}

The best upper bound is cubic and was built incrementally.
The classic Pin-Frankl bound $(n^3-n)/6-1$ ($n \ge 4$) \cite{Pin1972Utilisation}, which remained the best known for 33 years, follows as the sum of independent upper bounds for compressing a subset of the given size, i.e., for a subset $S \subseteq Q$ of size $|S|>1$, we bound the length of the shortest word such that $|\delta(S,w)|<|S|$ by $\binom{n-|S|+2}{2}$ \cite{Fr1982}.

The upper bound was slightly improved by mixing it with another compression technique through \emph{avoiding words}: ${\sim}0.1664 n^3 + \O(n^2)$ \cite{S2018ImprovingTheUpperBound}.
This bound was then improved to ${\sim}0.1654 n^3 + o(n^3)$ by finding the precise minima of certain functions \cite{Shitov2019}, and it currently remains the best known one.

\subsection{Avoiding Words}

The avoiding word problem is similar to synchronization: we want to find a word that ensures the automaton will \emph{not} be in a given state.
Avoiding words first appeared in~\cite{Tr2011ModifyingUpperBound} but under a different name.

A word $w \in \Sigma^*$ is \emph{avoiding} a state $q \in Q$ if $q \notin \delta(Q,w)$, i.e., whenever we apply $w$, we know that the automaton \emph{cannot} be in $q$.
A state $q \in Q$ is \emph{avoidable} if it admits an avoiding word.
In a synchronizing automaton with $n \ge 2$ states, all states are avoidable, since we can use a reset word, followed optionally by a letter that maps the single state to another one.

It may be useful to note that a word avoiding a state $q$ is equivalent to extending $Q \setminus \{q\}$ to $Q$.

The main open problem concerning avoiding words is bounding the \emph{avoiding threshold} -- the length of the shortest words avoiding a given state $q$, or generally, the maximum length such that every avoidable state can be avoided with a word no longer than that.

For example, the avoiding threshold of a state $q_i$ in the \Cerny automaton from~\Cref{fig:Cerny} is $i+1$, hence the avoiding threshold of the automaton is $n$.

\begin{conjecture}[Avoiding threshold \cite{KKS2016ExperimentsWithSynchronizingAutomata}]
In a synchronizing automaton with $n \ge 2$ states, every state has an avoiding word of length at most $2n-2$.
\end{conjecture}

However, there are only two known particular (over a minimal alphabet) examples reaching the bound $2n-2$, and $2n-3$ is the best known lower bound for all $n \ge 6$ \cite{FSV2021LowerBoundsOnAvoidingThresholds}.

Avoiding can also be considered from a given subset $S \subseteq Q$ instead of from $Q$.
Then a word $w$ is \emph{avoiding} a state $q$ \emph{from} $S$ if $q \notin \delta(S,w)$.

The best upper bound on avoiding threshold is quadratic, obtained through the following lemma:
\begin{lemma}[{\cite[Rephrased Lemma~1]{S2018ImprovingTheUpperBound}}]\label{lem:AvoidOrCompress}
Let $(Q,\Sigma,\delta)$ be an $n$-state automaton, $S \subseteq Q$ be non-empty, and $A \subsetneq S$ be non-empty.
If any state from $A$ is avoidable, then there exists a word $w$ of length at most $n-|A|$ satisfying either (1) $A \nsubseteq \delta(S,w)$ (avoid) or (2) $|\delta(S,w)| < |S|$ (compress).
\end{lemma}
By its iterative application, it follows that the upper bound for avoiding one state ($|A|=1$) is $(n-2)(n-1)+2$.

\begin{figure}[htb]\large\centering\begin{tikzpicture}[node distance=3cm,scale=0.8,every node/.style={transform shape}]
\tikzset{every loop/.style={min distance=.5cm,looseness=6}}
\node[state] (q0){$q_0$};
\node[state] [right of=q0] (q1){$q_1$};
\node[state] [right of=q1] (q2){$q_2$};
\node[state,,draw=none] [right of=q2] (qdots){$\ldots$};
\node[state] [right of=qdots] (qn-2){$q_{n-2}$};
\node[state] [below of=q2] (qn-1){$q_{n-1}$};

\draw[->,>=Stealth] (q0) to node[midway,above]{$a$} (q1);
\draw[->,>=Stealth] (q1) to node[midway,above]{$a$} (q2);
\draw[->,>=Stealth] (q2) to node[midway,above]{$a$} (qdots);
\draw[->,>=Stealth] (qdots) to node[midway,above]{$a$} (qn-2);
\draw[->,>=Stealth] (qn-2) to node[midway,below]{$a,b$} (qn-1);
\draw[->,>=Stealth] (qn-1) to[bend left=16] node[midway,below]{$b$} (q0);

\draw[->,>=Stealth] (q0) to node[midway,above]{$b$} (qn-1);
\draw[->,>=Stealth] (q1) to[bend right=4] node[above]{$b$} (qn-1);
\draw[->,>=Stealth] (q2) to node[left]{$b$} (qn-1);
\draw[->,>=Stealth] (qdots) to[bend left=4] node[midway,above]{$b$} (qn-1);

\draw[->,>=Stealth](qn-1) to[in=250,out=290,looseness=5] node[midway,below]{$a$} (qn-1);
\end{tikzpicture}
\caption{Automata where from $\{q_0,q_{n-1}\}$, the shortest word avoiding $q_{n-1}$ or compressing the pair is $a^{n-1}$.}\label{fig:LongAvoidingOrCompress}
\end{figure}
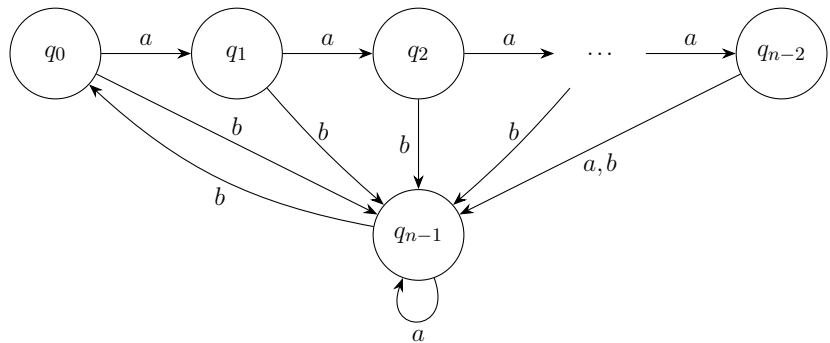

\Cref{fig:LongAvoidingOrCompress} shows that \Cref{lem:AvoidOrCompress} sometimes gives a tight upper bound, but we only know this for $|S|=2$ and $|A|=1$.
We can extend this for every $|S| \ge 2$ by adding isolated states, but that would not yield a synchronizing or strongly connected automaton.
But certainly not for all $|S|$, $|A|$; for instance, the shortest compressing words are shorter for large subsets $|S| \in n - \O(1)$.
It is not known whether the upper bound can be improved in other cases, which leads to the following open problem:
\begin{openproblem}
What are the tight upper bounds in~\Cref{lem:AvoidOrCompress}, depending on $|S|$ and $|A|$, in particular for a synchronizing and strongly connected automaton?
\end{openproblem}

Proving a linear upper bound on the length of avoiding words implies the upper bound $\frac{7}{48} n^3 + \O(n^2)$ on the reset threshold.
Providing good upper bounds for the more general concept of avoiding more than one state at once leads to bigger improvements.
More precisely, to achieve a subcubic upper bound on reset threshold, it would be sufficient to obtain a subquadratic upper bound on avoiding a subset of states of a logarithmic size.
For non-synchronizing automata, already for a pair of states we may need a word of length $\varTheta(n^2)$ to avoid it \cite[Theorem~9]{FSV2021LowerBoundsOnAvoidingThresholds}, and generally, $\varTheta(n^k)$ is a tight upper bound for avoiding $k$ states.
However, for synchronizing automata, using the technique might be possible, as we have the following strong conjecture:

\begin{conjecture}[{\cite[Rephrased Conjecture~15]{FSV2021LowerBoundsOnAvoidingThresholds}}]\label{con:k-avoidingThreshold}
For an $n$-state synchronizing automaton, for every subset $X$ of $k < n$ states, there exists an avoiding word for $X$ of length at most $\O(kn)$.
\end{conjecture}

The bound is smaller than needed (linear vs subquadratic), so there is some margin and weaker forms would work as well.

\subsection{Compressing a State with Another}

For a state $q \in Q$, we consider words whose action map $q$ to the same state together with some other state $p \ne q$.
A word $w \in \Sigma^*$ \emph{compresses a state with another} if $\delta(q,w) = \delta(p,w)$ for some $p \in Q \setminus \{q\}$.
We call the length of the shortest such words the \emph{compress-with-another} threshold.
This problem is very natural and closely related to the others, but so far largely overlooked and has not appeared in the published literature.

\begin{conjecture}[Compress-with-another threshold]\label{con:CompressWithAnother}
For a synchronizing strongly connected $n$-state automaton ($n \ge 2$), for every state, the compress-with-another threshold is at most $\O(n)$.
\end{conjecture}

The assumption that the automaton is synchronizing (and maybe strongly connected) is essential.
Without it, it is easy to construct automata with a state with a quadratic compress-with-another threshold.

\begin{proposition}\footnote{O{\.z}arowski 2017, personal communication}
For every $n$, there exists a (non-synchronizing and not strongly connected) $n$-state automaton with a state $q$ whose compress-with-another threshold is $\varTheta(n^2)$.
\end{proposition}
\begin{proof}[Proof idea]
We assume $n$ is odd; otherwise, we can add an isolated state.
Our set of states is $Q = \{q,s\} \cup C_1 \cup C_2$, where $|C_1|=\frac{n-3}{2}$ and $|C_2|=|C_1|+1$.
The alphabet is $\Sigma = \{a,b,c\}$.
The state $s$ is a sink with loops on all letters.

The action of $a$: It swaps $q$ with a state in $C_1$. It fixes the other states in $C_1$ and one state $x \in C_2$. The states from $C_2 \setminus \{x\}$ go to $s$.

The action of $b$: It fixes $q$ and $C_1$, maps $x$ to a state $y \in C_1$, and maps $C_2 \setminus \{x\}$ to $s$.

The action of $c$: It fixes $q$ and forms two cycles on $C_1$ and $C_2$, respectively.

The idea is as follows. The only way to compress $q$ with another state is by using $x \to y$ on $b$.
Starting with $a$ (then in $C_2$ there remains only $x$ in the image) followed by a suitable number of $c$, $q$ can be mapped to $y$ while $x$ is fixed; then finally $b$ compresses both states.
By choosing the index of $y$ in $C_1$, we can force a quadratic number of required $c$ applications.
\end{proof}

For the synchronizing case, compress-with-another threshold rarely exceeds $n$.
We know only one potential candidate for a larger value.

There exists a series with a state whose compress-with-another threshold is likely to be $\frac{4}{3}n$ \cite{D2018SynchronizingAutomataWithExtremalProperties}.
However, this fact seems to be difficult to prove, and hence we leave it as a conjecture for this construction.

\begin{figure}[!htb]
\begin{center}\includegraphics[scale=1]{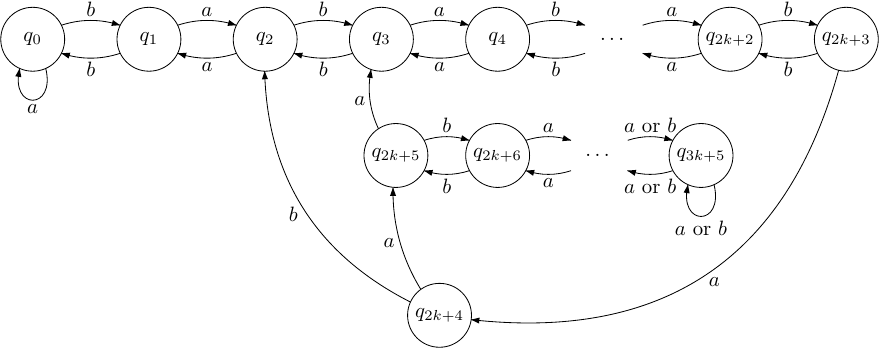}\end{center}
\caption{An automaton where state $q_0$ has possibly a large compress-with-another threshold.}\label{fig:CompressWithAnother}
\end{figure}

Let $k \ge 1$. Let $Q_{3k+6} = \{q_0,\ldots,q_{3k+5}\}$.
We define the automaton $\mathrsfs{A}_{k} = (Q_{3k+6},\{a,b\},\delta_{k})$ illustrated in~\Cref{fig:CompressWithAnother} as follows:\\
\begin{minipage}[t]{.49\textwidth}\raggedright
\indent $\delta_{k}(q_i,a) = \begin{cases}
q_0 & \text{if $i=0$}\\
q_{i+1} & \text{if $1 \le i \le 2k+3$, $2 \nmid i$}\\
q_{i-1} & \text{if $1 < i < 2k+3$, $2 \mid i$}\\
q_{i+1} & \text{if $2k+4 \le i < 3k+5$, $2 \mid i$}\\
q_{i-1} & \text{if $2k+7 \le i \le 3k+5$, $2 \nmid i$}\\
q_{3} & \text{if $i = 2k+5$}\\
q_{3k+5} & \text{if $ i =  3k+5$, $2 \mid i$}\\
\end{cases}$
\end{minipage}\hfill\noindent\begin{minipage}[t]{.49\textwidth}\raggedleft
\indent $\delta_{k}(q_i,b) = \begin{cases}
q_{i+1} & \text{if $0 \le i < 2k+3$, $2 \mid i$}\\
q_{i-1} & \text{if $0 < i \le 2k+3$, $2 \nmid i$}\\
q_{i+1} & \text{if $2k+5 \le i < 3k+5$, $2 \nmid i$}\\
q_{i-1} & \text{if $2k+6 \le i \le 3k+5$, $2 \mid i$}\\
q_{2} & \text{if $i = 2k+4$}\\
q_{3k+5} & \text{if $ i =  3k+5$, $2 \nmid i$}\\
\end{cases}$\\
\end{minipage}

The automaton $\mathrsfs{A}_k$ is obviously strongly connected.
However, for the other properties we have only a conjecture:

\begin{conjecture}
For $k \ge 1$, the automaton $\mathrsfs{A}_k$ with $n=3k+6$ states is synchronizing , and the shortest word that compresses $\{q_0,q\}$ for some $q \in Q \setminus \{q_0\}$ is $(ba)^{k+2} (ab)^{k+2}$ of length $\frac{4}{3} n$.
\end{conjecture}

The series remains unproven, as proving its properties is difficult due to the lack of regularity in the behavior of this series.
It seems to be difficult even to prove that these automata are synchronizing.
The conjecture has been verified computationally for a small number of states $n \le 21$.

If compress-with-another threshold is at most linear (\Cref{con:CompressWithAnother}), then it brings a very simple proof of an $\O(n^2)$ upper bound on the avoiding threshold.
Here it is the idea:
To avoid a state $q$, we find a word $w$ compressing $\{q,p\}$ for some other state $p$.
If the word avoids $q$, we are done.
Otherwise, we apply a shortest word (of length $\le n-1$) whose action maps $q$ to $p$.
If by that, $q$ has not been avoided, then we have both $q,p$ in the image, so the procedure can be repeated.
We repeat it at most $n-1$ times until eventually a singleton is reached, each time appending a word of length at most $\O(n)$.

The problem is also related to synchronization of any subset (Subsection~2.5).

\subsection{Synchronization of a Given Subset}

Another natural generalization concerns the length of words needed to synchronize a given subset. For $S=Q$, this is the \Cerny conjecture, and for $|S|$ we know the tight upper bound $n(n-1)/2$.

In general, the following conjecture has been stated by Cardoso:
\begin{conjecture}[{\cite[Conjecture~5.1.1]{Cardoso2014PhD}}]
For a synchronizing $n$-state automaton, every subset $S$ of states can be synchronized with a word of length at most:
\[ (n-1)^2 - \left\lceil\frac{n-|S|}{|S|}\right\rceil (2n - |S| \left\lceil\frac{n}{|S|}\right\rceil - 1) .\]
\end{conjecture}
It has been shown that the \Cerny automata have subsets meeting this bound for every $|S|$.

A better upper bound than the trivial one (obtained from summing the upper bounds on compression) for any $|S| \ge 3$ would improve the known upper bound on the reset threshold.

\subsection{Synchronization of Any \texorpdfstring{$k$-subset}{k-subset}}

The counterpart variation of the previous problem is synchronization of any subset of $k$ states.
Here the subset is not fixed: we ask for a word that is synchronizing for some subset with $k$ states, so this is the question about synchronizing the easiest one.

For $k=2$, this is trivial, as any non-permutational letter does the job.
For $k=3$ and a synchronizing automaton, this was studied under the name \emph{triple rendezvous time} \cite{GonzeJungers2016OnSynchronizingProbabilityFunction}.
The best upper bound is quadratic but better than the trivial one, derived using linear-algebraic methods.
\begin{theorem}[{\cite{GonzeJungers2016OnSynchronizingProbabilityFunction}}]
In a synchronizing $n$-state automaton, the shortest words synchronizing a $3$-subset have length at most $0.1545 n^2 + \O(n)$.
Furthermore, the bound $n+3$ is attainable.
\end{theorem}
For $k=4,5$, there are better bounds than the trivial one \cite{BJ2022SynchronizingKSets}.
For non-synchronizing automata, the tight bound is $\varTheta(n^k)$ \cite{BJ2022SynchronizingKSets}, thus the situation is very similar to that of the avoiding words.

The case $k=3$ is very similar to the previous problem of compressing a state with another.
Indeed, we necessarily first use a non-permutational letter $a \in \Sigma$ to compress a pair of states, and then we need to compress the state in their image of this pair with some other state -- the only difference is that exactly one state is missing in the image $\delta(Q,a)$ so it cannot be that other state.

We can risk the following:
\begin{conjecture}[cf.\ {\cite[Conjecture~20]{BJ2022SynchronizingKSets}}]
In a synchronizing $n$-state automaton, the shortest words synchronizing any $k$-subset ($2 \le k \le n)$ have length at most $\O(kn)$.
\end{conjecture}
Note that it is very similar to our \Cref{con:k-avoidingThreshold} about avoiding a $k$-subset.

In fact, compression a state with another is somewhere in between both variants of subsets synchronization.
They can be combined into one bigger open problem:
\begin{openproblem}
For a subset $S \subseteq Q$ and an integer $0 \le k \le |Q|-|S|$, find upper bounds on the length of the shortest words $w$ such that $|\delta(S \cup S',w)|=1$ for some $S' \subseteq Q \setminus S$ of size $|S'|=k$.
\end{openproblem}
For compressing a state with another, we set $|S|=\{q\}$ and $k=1$.
Of course, it could be generalized further to include avoiding and compression to a given subset size.

\subsection{Deciding the Synchronizability}

The best-known algorithm for deciding whether a given automaton is synchronizing works in $\O(|\Sigma|\cdot n^2)$ \cite{Ep1990}.
It is a simple algorithm which builds a tree on pairs of states in a backward way to check whether from every pair of states we can reach a singleton.
However, there are no clues whether this is the best possible, in particular regarding the quadratic factor in $n$.
For the simplicity of the following question, let us assume a fixed alphabet.

\begin{openproblem}\label{pro:CheckingIfSynchronizing}
Given $\mathrsfs{A}$ under a fixed alphabet, is there a faster algorithm than $\O(n^2)$ checking whether it is synchronizing?
\end{openproblem}

A possible way to argue for the impossibility of a (significant enough) improvement could be a conditional lower bound, assuming some established hypothesis from computational complexity.
An example of such a problem is the nonempty intersection of languages recognized by two given DFAs:
If two binary tree-shaped DFAs can be solved in $\O(n^{2-\varepsilon})$ time, for any $\varepsilon>0$, then the Strong Exponential Time Hyphothesis (SETH) is false \cite{OliveiraWehar2020FiniteAutomataNonemptinessIntersection} (there are even stronger results than contradicting SETH).

Recall that the synchronizability is equivalent to the condition that all pairs of states are compressible.
Curiously, if one asks whether a \emph{given} pair of states is compressible, the problem cannot be solved in $\O(n^{2-\varepsilon})$.
The reduction is the same as used to show PSPACE-completeness of subset synchronization \cite{Rystsov1983PolynomialCompleteProblems}, which also reduces from the nonempty intersection problem with an unrestricted number of given DFAs.
It is enough to apply it for two input automata and combine the results.
Since the given DFAs are binary and we add only one extra letter to the alphabet, we can also binarize the final construction at the cost of doubling the number of states.

\begin{proposition}
Given an automaton $\mathrsfs{A}=(Q,\Sigma,\delta)$ over a fixed alphabet and a pair of states $p,q \in Q$, there is no algorithm checking whether $\{p,q\}$ is compressible in time $\O(n^{2-\varepsilon})$, for every $\varepsilon>0$, assuming SETH.
\end{proposition}

In~\cite{KR2025ComplexityOfReachabilityProblems}, it is shown that deciding if an automaton is synchronizing is NL-complete, even if the automaton is binary and strongly connected.
The same holds for checking if a given pair of states is compressible. Unfortunately, the constructions from these results do not seem to help with resolving \Cref{pro:CheckingIfSynchronizing}.

We note that for a random binary automaton, there exists an algorithm deciding if it is synchronizing in $\O(n)$ time \cite{Berlinkov2016OnTheProbabilityToBeSynchronizable}.
By \emph{random}, we mean taking each transition uniformly and independently at random, so each outgoing transition has exactly $n$ choices, and there are $n^{|\Sigma|\cdot n}$ such random automata.
This algorithm works because a random binary automaton is synchronizing with high probability, and its synchronizability can also be confirmed in linear time w.h.p., whereas the remaining cases are treated with the standard quadratic algorithm as a fallback.

\subsection{Average Reset Threshold}

The question about the expected reset threshold of a random $n$-state automaton has been widely studied, e.g., \cite{BS2016AlgebraicSynchronizationCriterionAndComputingResetWords,CP2025ShortSynchronizingWordsForRandomAutomata,Hi1988,KKS2015ComputingTheShortestResetWords,Nicaud2016FastSynchronizationOfRandomAutomata,ST2011}.
The binary case is the most fundamental.
In the unary case, the probability of being synchronizing is exactly $1/n$ (we count the rooted trees on $n$ labeled vertices -- there are $n^{n-1}$ of them versus $n^n$ of all possible letter actions).

However, most efforts concern an easier question: an upper bound on the average reset threshold that holds with high probability.
Of course, the questions are not equivalent, as a certain upper bound can hold with high probability, but the average reset threshold in the remaining cases might be large enough to substantially increase the total average.
Yet, the function is believed to be the same.

The strongest result has been recently shown by Chapuy and Perarnau: 
\begin{theorem}[{\cite{CP2025ShortSynchronizingWordsForRandomAutomata}}]
A random binary automaton has a reset threshold in $\O(\sqrt{n} \log n)$ with high probability.
\end{theorem}

\begin{figure}[!htb]\centering\includegraphics[scale=.6]{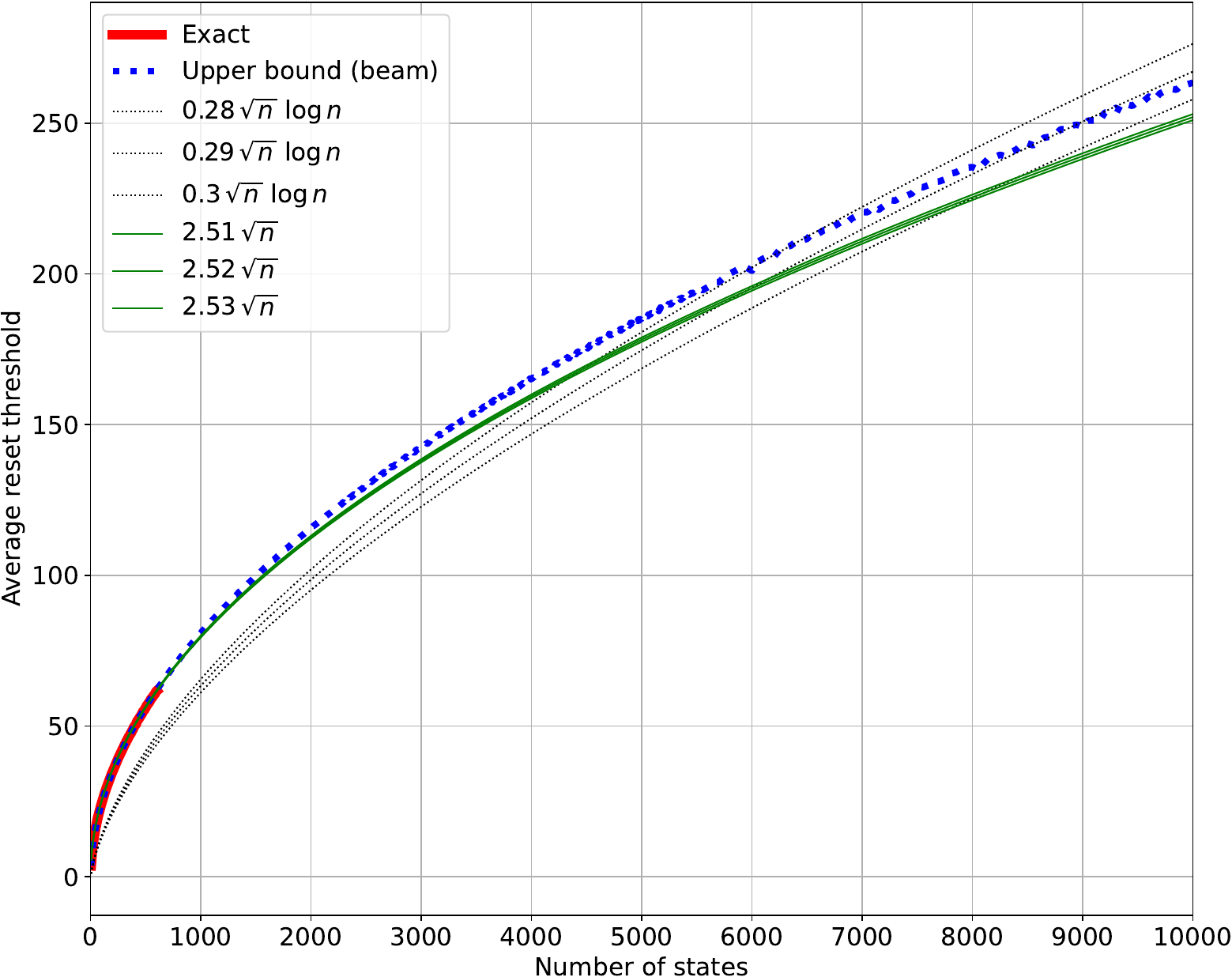}
\caption{The average reset threshold of random binary automata from experiments compared with $\varTheta(\sqrt{n})$ and $\varTheta(\sqrt{n} \log n)$ functions.}\label{fig:AvgRTFull}
\end{figure}
\begin{figure}[!htb]\centering\includegraphics[scale=.6]{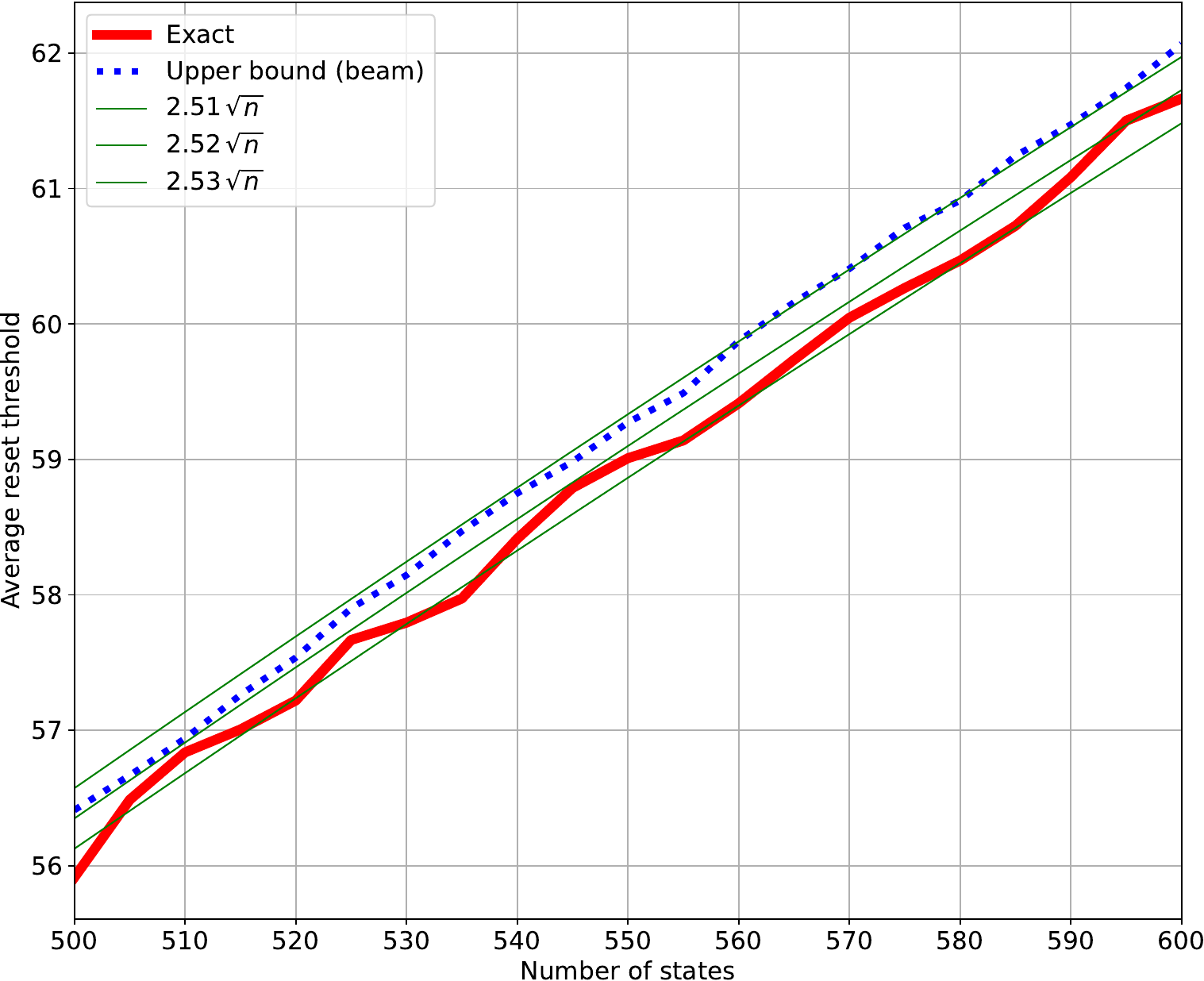}
\caption{The average reset threshold of random binary automata from experiments compared with $\varTheta(\sqrt{n})$ functions, zoomed to the range $n=500,\ldots,600$.}\label{fig:AvgRTZoom}
\end{figure}

There is experimental evidence suggesting the bound is not tight.
\Cref{fig:AvgRTFull} and \Cref{fig:AvgRTZoom} show the latest experimental results compared with several functions.
The results were obtained using extensive computation with the fastest known algorithm computing the exact reset threshold ($n=5,10,15,\ldots,600$) and a good heuristic algorithm providing an upper bound (up to $n \le 10,000$) (these follow from an extended study using the algorithm from~\cite{SZ2022ImprovedAlgorithmForFindingTheShortestSynchronizingWords}); for each $n$, there are at least $1,000$ automata sampled).

The following conjectures are discussed and seem to be supported by the experiments:

\begin{conjecture}[Average reset threshold (weak)]
A uniformly random binary automaton has reset threshold in $\varTheta(n^{1/2})$ with high probability.
\end{conjecture}

\begin{conjecture}[Average reset threshold (strong)]
The expected reset threshold of a uniformly random binary automaton is in $\varTheta(n^{1/2})$.
\end{conjecture}

There is a big gap between these two problems, as the strong one is much more difficult to deal with theoretically.
For the strong variant, we have only an upper bound of $\O(n^{1+\varepsilon})$, which follows from~\cite{CP2025ShortSynchronizingWordsForRandomAutomata} in combination with the linear-algebraic method \cite{BS2016AlgebraicSynchronizationCriterionAndComputingResetWords} to include the remaining cases of automata:
the high probability in \cite{CP2025ShortSynchronizingWordsForRandomAutomata} is $1-\tilde{\O}(1/\sqrt{n})$, which combined with the general cubic upper bound gives only $\O(n^{5/4+o(1)})$, or $\O(n^{1+o(1)})$ assuming a quadratic upper bound on the reset threshold.

On the other hand, the known lower bound is only $n^{1/3 - \varepsilon}$ for all $\varepsilon > 0$ and has a simple proof (also in~\cite{CP2025ShortSynchronizingWordsForRandomAutomata}).

\subsection{Improving Linear-Algebraic Methods}

The linear-algebraic methods, which originate from early works of Pin (\cite{Pin1972Utilisation}), have demonstrated spectacular successes in proving better upper bounds on reset threshold for particular subclasses of automata (e.g., \cite{BBP2011QuadraticUpperBoundInOneCluster,BS2016AlgebraicSynchronizationCriterionAndComputingResetWords,Dubuc1998,GonzeJungers2016OnSynchronizingProbabilityFunction,Kari2003Eulerian,Pin1972Utilisation,Steinberg2011AveragingTrick,Steinberg2011OneClusterPrime}).

Furthermore, it is currently applied behind the best cubic upper bound as well (\cite{Shitov2019,S2018ImprovingTheUpperBound}, and in~\cite{Fr1982} for the old Pin-Frankl bound, although a bit differently).
Recently, using the technique, a quadratic upper bound ($2n^2 -7n + 7$) has been shown for the class of synchronizing automata with a transitive permutation group and where all letters have rank at least $n-1$ \cite{Zhu2024QuadraticUpperBoundTransitivePermutationGroup}.
However, the method by itself does not yield tight bounds in most cases.

\subsubsection{Introductory Definitions}

We associate a natural linear structure with an automaton $\mathrsfs{A}$.
By $\mathbb{R}^n$ we denote the real\footnote{For (almost) all proofs, it is enough to use rational numbers $\mathbb{Q}$ as well.} $n$-dimensional linear space of row vectors.
Without loss of generality, we assume that $Q=\{1,2,\dots,n\}$.
For a subset $S \subseteq Q$, $[S] \in \mathbb{R}^n$ is the \emph{characteristic 0-1 vector} of $S$, where the $q$-th entry is $1$ if $q \in S$, and $0$, otherwise.

Analogously, a word $w \in \Sigma^*$ corresponds to a linear transformation of $\mathbb{R}^n$.
By $[w]$, we denote the matrix of this transformation in the standard basis $[1],\ldots,[n]$ of $\mathbb{R}^n$.
Clearly, the matrix $[w]$ has exactly one non-zero entry in each row, because the automaton is deterministic.
In particular, $[w]$ is \emph{row stochastic}, i.e., the sum of entries in each row is equal to $1$, hence multiplying by $[w]$ preserves the sum of the entries of a vector.
Then we have $[uv]=[u][v]$ for every two words $u,v \in \Sigma^{*}$. 

Given a subset $S \subseteq Q$ and a word $w \in \Sigma^*$, the $[S][w]$ is the \emph{cumulative image}, which is equal to $[\delta(S,w)]$ if and only if no two of states in $S$ were compressed by $w$.
However, a $q$-th entry of $[S][w]$ is non-zero if and only if $q \in \delta(S,w)$.

On the other hand, when considering preimages, we always have $[S][w]^\mathrm{T} = [\delta^{-1}(S,w)]$, where $[w]^\mathrm{T}$ is the transpose of $[w]$.

When $X$ is a set of vectors, by $\mathcal{V} = \lspan(X)$ we denote the linear subspace of $\mathbb{R}^n$ generated by the vectors from $X$ as a linear combination.
Let $\dim(V)$ denote the \emph{dimension} of $V$, and also for a set of vectors, $\dim(X)$ is $\dim(\lspan(X))$.

\subsubsection{Ascending Chain of Subspaces}

Let $\Sigma^i$ and $\Sigma^{\le i}$ denote all words over $\Sigma$ of length equal $i$ and at most $i$, respectively.

The common way to utilize linear algebra for synchronizing automata is to build a chain of subspaces starting from a subset $S$:
\[ L_i = \lspan(\{ [S][w] \mid w \in \Sigma^{\le i} \}),\text{ for $i=0,1,\ldots$}, \]
or, when working with preimages:
\[ L_i = \lspan(\{ [S][w]^\mathrm{T} \mid w \in \Sigma^{\le i}) \},\text{ for $i=0,1,\ldots$}. \]
The key argument is that if the next subspace (generated using longer words) has the same dimension as the previous one, then all the further subspaces stabilize and are the same.
So if reaching a certain dimension implies that there is a word of a desired property, then we have a linear bound on its length.

To work on a concrete example, we recall the proof of the upper bound on the reset threshold of an Eulerian automaton \cite{Kari2003Eulerian}, for its relative simplicity.
The utilized property of an Eulerian automaton is that, for every non-empty subset $S \subsetneq Q$, if a word $w \in \Sigma^i$ is such that $|\delta^{-1}(S,w)| < |S|$, then there exists a word $w' \in \Sigma^i$ extending $S$: $|\delta^{-1}(S,w)| > |S|$.

For a non-zero 0-1 vector $v$, let $\mathrm{norm_0}(v)$ be this vector normalized in the following way.
Let $|v|$ be the sum of the entries of $v$, i.e., the number of entries with $1$ in this case.
The entries with $1$ are replaced with $1/|v|$, and the entries with $0$ are replaced with $-1/(n-|v|)$.
The resulting vector $\mathrm{norm_0}(v)$ has the sum of entries equal to $0$.
Also, for a set of non-zero 0-1 vectors $V$, let $\mathrm{norm_0}(V)$ be the set of the normalized vectors from $V$.

For a non-empty subset $S \subsetneq Q$, we build the chain of subspaces:
\[ L_i = \lspan\left(\mathrm{norm_0}([S]) [w]^\mathrm{T} \mid w \in \Sigma^{\le i}\right),\text{ for $i=0,1,\ldots$}. \]
If for an $i$, there are vectors in $L_i$ with sum $\ne 0$, this means the number of positive and negative entries is different than in $\mathrm{norm_0}([S])$, hence there exists some word $w \in \Sigma^{n-1}$ such that $|\delta^{-1}(S,w)| \ne |S|$.
By the Eulerian property, there exists an extending word for $S$ of length at most $i$.
Since the maximum possible dimension is $n$ and $L_0$ has dimension $1$, we have $i \le n-1$.

Applying the argument iteratively (at most $(n-2)$ times), starting from $\delta^{-1}(\{q\},a)$ for a letter that extends $\{q\}$, we end up with the upper bound $1+(n-2)(n-1)$ on the length of a reset word, i.e., $\delta^{-1}(\{q\},aw) = Q$, where $|w| = (n-2)(n-1)$.
This is the state-of-the-art upper bound on the reset threshold of an Eulerian automaton, whereas the best known lower bound is $\lfloor \frac{n^2-3}{2}\rfloor$ (conjectured to be tight) \cite{SV2016ExtremalEulerian}.

\subsubsection{Generalization and an Improvement}

Instead of always starting independently from a new subset $S$, we can build our chains of subspaces starting from more vectors (more initial subsets) at once.
In this way, we can increase the initial dimension, which will let the chains reach the maximum dimension faster.
As a drawback, this leads to conditional bounds: we will obtain an extending word for one of the initial subsets, but will not know for which one.

The idea of starting from more subsets has already appeared in the proofs for one-cluster automata, e.g., \cite{Steinberg2011AveragingTrick,Steinberg2011OneClusterPrime}, yet this is different (and independent) from our proposal, as iterations are applied to one subset that emerges from the previous one as usual.

First, for $S_1,\ldots,S_j$, let us define:
\[ L_0 = \lspan\left(\mathrm{norm_0}([S_1],\ldots,[S_j])\right); \quad L_i = \lspan\left(L_0 [w]^\mathrm{T} \mid w \in \Sigma^{\le i}\right),\text{ for $i=0,1,\ldots$} .\]
Then it gives an extending word $w$ for some of $S_1,\ldots,S_j$ of length at most $i \le n-\dim(L_0)$.

Now, to use it in the construction of a reset word:
Starting from $S_1 = \delta^{-1}(\{q\},a)$, suppose that we iteratively find the \emph{shortest} extending words $w_1,\ldots,w_k$ for the current subsets $S_1,\ldots,S_k$, ending up with $Q$, where $k \le n-2$.
We want to bound the (sum of the) lengths of these extending words.
For $S_1,\ldots,S_k$, we have the upper bound $n-\dim(\lspan(\mathrm{norm_0}([S_1],\ldots,[S_k])))$ on the length of one of them.
We remove the extended subset and use the argument for the set of $k-1$ vectors, repeating it until no subset is left.

First, we relate the dimension of a normalized vectors' subspace with the dimension of the 0-1 vectors' subspace, which will simplify our dimension consideration to only the latter case.
Let $\mathbf{1}^n$ be the vector with all $1$s.

\begin{lemma}\label{lem:Norm0VS01}
Let $X$ be a set of 0-1 vectors in $\mathbb{R}^n$ that contains $\mathbf{1}^n$.
Then $\dim(X) = \dim(\mathrm{norm_0}(X)) + 1$.
\end{lemma}
\begin{proof}
To show $\dim(X) \ge \dim(\mathrm{norm_0}(X))+1$, we observe that all the normalized vectors are in $\lspan(X)$: for a vector $\mathrm{norm_0}(v)$, which has only $a$ and $-b$ entries, we multiply $v$ by $a+b$ and subtract $\mathbf{1}^n$ multiplied by $b$.
Additionally, the normalized vectors cannot generate vectors with non-zero sum, so clearly $X$ generates a larger subspace.

To show $\dim(X) \le \dim(\mathrm{norm_0}(X))+1$, we observe that adding $\mathbf{1}^n$ to the set $\mathrm{norm_0}{X}$ allows us to generate all the vectors from $X$.
For the normalized vectors, we apply the inverse transform to the above: to a vector with $a$ and $-b$ entries, we add $\mathbf{1}^n$ multiplied by $b$ and then divide by $a+b$.
\end{proof}

The improvement depends on the dimension of the subspaces generated by the vectors of the subsets.
Unfortunately, in the general case, the initial dimension can be logarithmic in the number of vectors, but it can be linear in the best case.
We show an elementary combinatorial proof of the logarithmic lower bound.

\begin{lemma}\label{lem:DimLog}
Let $X$ be a set of non-zero 0-1 vectors in $\mathbb{R}^n$.
Then $\dim(X) \ge \lceil \log_2 |X| \rceil$.
\end{lemma}
\begin{proof}
We construct a binary tree as follows.
For a set $X'$ of 0-1 vectors, we find the first distinguishing position, which is the smallest index such that there is a vector in $X'$ that has $0$ and another that has $1$ on this position.
Then we split the set into two, depending on whether they have $0$ or $1$ on this position.
For both subsets, we construct the subtree recursively.
If there is no distinguishing position, then we already have only one vector in the set, so we end with a leaf.
Since we have exactly $|X|$ leaves and the tree is binary, the height of the tree is at least $\lceil \log_2 |X|\rceil$ (the root is at level $0$).

We show that we can find a subset $Y \subseteq X$ of linearly independent vectors by following the longest path in the tree from the root to the leaf.
We start from the leaf, taking the stored vector as the first one into our set $Y$.
Being at a node $x$, we go up to its parent $x'$.
Note that all vectors taken so far have the same value on the distinguishing position of $x'$.
We follow the second edge from $x'$ up to any leaf of that subtree, taking a vector $v$ stored there.
Then $v$ has a different value at the distinguishing position of $x'$ than all the vectors in $Y$ taken so far.
\begin{enumerate}
\item If $v$ has $1$ at the position of $x$', then it is linearly independent with $Y$, since all vectors in $Y$ have $0$ on this position.
So we take $v$ to $Y$.
\item If $v$ has $0$ at the position of $x'$, then we have two subcases.
If all vectors in $Y$ and $v$ have $0$ on all earlier positions (the path from the root to $x'$ is labeled only by $0$s), then we ignore $v$ and continue.
Note that this subcase can happen only once, as all further cases must be (1).
Otherwise, all vectors in $Y$ and $v$ have $1$ on some common (earlier) position, so the vectors in $Y$ preserve the sum on that position and on the position of $x'$, whereas $v$ breaks it.
Hence, $v$ is linearly independent with $Y$ and we take it.
\end{enumerate}

Altogether, we visit $\lceil \log_2 k\rceil + 1$ nodes, where at most once we do not add a vector.
So we have at least $\lceil \log_2 k\rceil$ linearly independent vectors.
\end{proof}

For our purposes, from~\Cref{lem:Norm0VS01}, we can simply use $\dim(\mathrm{norm_0}(X)) \le \dim(X) - 1$ and consider just raw vectors $[S_1],\ldots,[S_k]$.
Consider the case of one initial subset separately.
From~\Cref{lem:DimLog}, our upper bounds on the length of the shortest extending words become:
\[ kn - \left(1 + \sum_{i=2}^k (\lceil \log_2 i \rceil - 1)\right) = kn + k-2 - \sum_{i=2}^k \lceil \log_2 i \rceil. \]

Simplifying by an analytic upper bound, we can already slightly improve the upper bound on the reset threshold of an Eulerian automaton:
\begin{corollary}
The reset threshold of a synchronizing $n$-state Eulerian automaton is at most
\[ n^2 - 2n + n-3 - \frac{(n-2) \ln (n-2)}{\ln 2} + \frac{(n-3)}{\ln 2} = n^2 - n \log_2 n + \O(n) .\]
\end{corollary}
\begin{proof}
\[ \sum_{i=2}^k \lceil \log_2 i \rceil \ge \sum_{i=1}^k \log_2 i \ge \int_{i=1}^k \log_2 i = \frac{k \ln k - k + 1}{\ln 2} .\]
Hence, our upper bound is bounded above:
$kn + k-2 - \sum_{i=2}^k \lceil \log_2 i \rceil \le kn + k-2 - \frac{k \ln k - k + 1}{\ln 2}$.
Since this is an increasing function, by setting the worst case $k=n-2$, we get:
\[ n^2 - 2n + n-4 - \frac{(n-2) \ln (n-2)}{\ln 2} + \frac{(n-3)}{\ln 2} = n^2 - n \log_2 n + \O(n).\]
For the final bound, we need to add $1$ for the initial extending letter.
\end{proof}

Of course, this is a little improvement, which merely highlights the non-optimality of the technique and its limits.
Yet, it can be employed in most of the proofs of this kind, in both image and preimage variants, weighted, etc.
One could go through all of them, decreasing the bounds by subtracting about $n \log_2 n$.
For instance, it leads to an upper bound on the avoiding threshold of order $n^2 - n \log_2 n + \O(n)$.

Unfortunately, the logarithmic upper bound on the initial dimension is tight in general: trivially, $2^k$ vectors that cover all 0-1 combinations on $k$ positions generate the dimension $k$.
The bound is also tight if all the vectors have pairwise different sums (as in our case).

\begin{lemma}
For all $k \le n$, there is a set $X$ of $k$ 0-1 vectors in $\mathbb{R}^n$ with sums of entries $1,2,\ldots,k$ such that $\dim(X) = \lceil \log_2 k \rceil$ (when $k$ is not a power of $2$) or $\dim(X) = 1 + \log_2 k$ (otherwise).
\end{lemma}
\begin{proof}
Let $j$ be the largest integer such that $2^j-1 \le k$, hence $j = \lfloor\log_2(k+1)\rfloor$.
We set $j$ vectors with the sum of the entries equal to $2^0,2^1,\ldots,2^{j-1}$, respectively, and where the $1$s appear at pairwise disjoint positions.
The above bound ensures that the number of needed positions is at most $k$.
We need to add the remaining $k-j$ vectors: note that each natural number up to $2^j-1$ is expressible as the sum of our powers of $2$, hence we can take vectors that are generated from our $j$ generating vectors.

If $2^j-1 = k$, then we already have $k$ vectors and the dimension is equal to $j$.
Then $k+1$ is a power of $2$, thus $j = \lfloor\log_2(k+1)\rfloor = \log_2(k+1) = \lceil\log_2 k\rceil$.

If $2^j-1 < k$, then for the vectors with $2^j,\ldots,k$ $1$s, we need one more generating vector with $1$s on $k-2^j+1$ so far unused positions.
Then for each $2^j,\ldots,k$, we can find a vector with that sum which can be generated using the additional generating vector and a subset of our $j$ vectors.
So we have $\lfloor\log_2(k+1)\rfloor + 1$ generating vectors.
Since $k+1$ is not a power of $2$, this is equal to $\lfloor\log_2 k\rfloor + 1$.
If $k$ is not a power of $2$, we end with $\lceil\log_2 k\rceil$.
Otherwise, we end with $1 + \log_2 k$.
\end{proof}

However, it might lead to better upper bounds, as the logarithmic worst-case is unlikely to happen with vectors in the context of an automaton.
The following is an informal question on that:

\begin{openproblem}
What are the tight upper bounds for the initial dimensions of the subspace of vectors of subsets appearing from an iterative application of the linear-algebraic method for constructing a word?
\end{openproblem}

\subsection*{Acknowledgments}

I thank Mikhail Berlinkov, Paweł Gawrychowski, Robert Ferens, Vladimir Gusev, Jarkko Kari, Igor Rystsov, and Andrew Ryzhikov for discussions on the described problems, which resulted in some of the observations.
\bibliographystyle{eptcs}
\bibliography{bibliography}
\end{document}